\documentclass[12pt]{article}
\usepackage{amssymb,amsmath,amsthm,graphicx,ulem}

\usepackage{graphicx}
\usepackage{epsfig}
\usepackage{amsmath, amssymb, amsthm, amssymb, amsfonts}
\usepackage[usenames]{color}
\usepackage[letterpaper,left=2.cm,right=2.cm,top=2.5cm,bottom=2.5cm]{geometry}
\usepackage{float}
\usepackage{xcolor}
\usepackage{hyperref}
\hypersetup{
    colorlinks = true,
    linkcolor = {black},
    urlcolor = {blue},
    citecolor = {red},
}
\usepackage{bm}
\usepackage[utf8]{inputenc}
\usepackage[english]{babel}
\usepackage{url}
\usepackage{authblk}
\usepackage{cite}
\usepackage{caption, subcaption}
\usepackage{listings}
\usepackage{comment}
\usepackage{dsfont}
\usepackage{algorithm, algpseudocode}

\DeclareSymbolFont{bbold}{U}{bbold}{m}{n}
\DeclareSymbolFontAlphabet{\mathbbold}{bbold}
\newtheorem{prop}{Proposition}

\DeclareFontFamily{U}{mathx}{\hyphenchar\font45}
\DeclareFontShape{U}{mathx}{m}{n}{<-> mathx10}{}
\DeclareSymbolFont{mathx}{U}{mathx}{m}{n}
\DeclareMathAccent{\widebar}{0}{mathx}{"73}

\begin{document}
 
\allowdisplaybreaks
\normalem
\title{The Hierarchical Parity Model}

\author[ ]{Gavin~S.~Hartnett\footnote{hartnett@rand.org}}

\affil[ ]{RAND Corporation, Santa Monica, CA 90401, USA}

\date{}
\setcounter{Maxaffil}{0}
\renewcommand\Affilfont{\itshape\small}

\maketitle 

\begin{abstract}
Hierarchical spin-glasses are Ising spin models defined by recursively coupling together two equally-sized sub-systems. In this work a new hierarchical spin system is introduced wherein the sub-systems are recursively coupled together through the parity of their spins. Exact Renormalization Group recursion equations for many correlators may be derived for this Hierarchical Parity Model, even for completely general couplings. Moreover, the model is computationally tractable in that $O(N)$ algorithms exist for the computation of the partition function and ground state energy. In the special case where the couplings are all equal, the model is shown to exhibit a thermal phase transition. 
\end{abstract}

\baselineskip16pt

\section{Introduction}
The defining property of hierarchical spin-glasses is that the Hamiltonian may be defined recursively in terms of the Hamiltonians of the two sub-systems formed by dividing the spins into two equally sized groups:
\begin{equation}
    \label{eq:HM}
    H_{n}(s) = H_{n-1}(s_L) + H_{n-1}(s_R) + \epsilon_{n}(s) \,.
\end{equation}
Here $s$ denotes all the Ising spins $s_1, ..., s_{N}$, with $N=2^n$ constrained to be a power of 2, $s_L$ denotes first half of spins, $s_1, ..., s_{2^{n-1}}$ and $s_R$ the second half, $s_{2^{n-1}+1}, ..., s_{2^n}$. The $L,R$ notation stands for ``left'' and ``right'', and corresponds to the fact that the recursion structure is naturally associated with a balanced binary tree. The two sub-systems are coupled together through an interaction term $\epsilon_n(s)$ which specifies the model. 

Dyson was the first to study models of this form, introducing a ferromagnetic model known as the Dyson Hierarchical Model (DHM) \cite{dyson1969existence}, for which the interaction term is
\begin{equation*}
    \epsilon_n(s) = - J \, C^{n} \left(\frac{1}{2^{n}} \sum_{i=1}^{2^{n}} s_i \right)^2 \,.
\end{equation*} 
Here $J>0$ is a ferromagnetic coupling, and $C$ controls how the interaction strength scales with system size. At each level in the recursion, the two sub-systems are coupled together via the square of the magnetization of the combined system. The hierarchical structure greatly aids the analysis of this model, and in particular leads to the key result that the Wilsonian Renormalization Group (RG) equations are exact when applied to this model. 

Due to the exactness of the RG equations, the DHM and other hierarchical models are useful toy models for studying and further developing the Renormalization Group (RG). In particular, it has proven especially difficult to develop a RG theory for spin-glasses, especially non-mean field systems. As a result, spin-glass versions of hierarchical models have been introduced and used to develop a theory of renormalization for these systems which might extend to other, non-hierarchical spin-glasses.\footnote{For an excellent review of this subject, see \cite{castellana2013renormalization}.} For example, the Hierarchical Random Energy Model (HREM) \cite{castellana2010hierarchical, castellana2011real} is defined by taking the couplings $\epsilon_n(s)$ to be independent and identically distributed random variables (with no dependence on the spin-configuration). Another well-studied example is the Hierarchical Edwards-Anderson (HEA) model \cite{franz2009overlap, castellana2010renormalization, castellana2011renormalization, castellana2011real}, for which the interaction term is 
\begin{equation*}
    \epsilon_{n}(s) = - \frac{C^{2n}}{2^n} \sum_{i < j = 1}^{2^{n}} J_{ij} s_i s_j \,,
\end{equation*}
where the couplings are standard normal random variables, i.e. $J_{ij} \sim \mathcal{N}(0,1)$.

In this work, a new hierarchical Ising spin model is introduced for which the coupling term $\epsilon_n(s)$ is simply the parity of the combined spins from the two sub-systems. The motivation for considering such an interaction is to develop a toy model of a spin system which is capable of exhibiting geometric frustration and large degeneracies, and which nonetheless exhibits a high degree of analytic and computational tractability. However, more work is needed to determine whether this model exhibits a proper spin-glass phase, and therefore we will refer to it as a spin system, as opposed to a spin-glass. 

This paper is organized as follows. In Section~\ref{sec:model} we introduce the model and derive its key recursion relations. In Section~\ref{sec:complexity} we discuss the computational complexity of the model, and show that $O(N)$ algorithms exist for computing both the partition function and the ground state. In Section~\ref{sec:widthsymmetric} we consider a special case of the model in which the couplings are width-symmetric (i.e., they are uniform within a given level of the hierarchy). When the couplings are furthermore taken to be equal across all levels of the hierarchy, the model is shown to exhibit a thermal phase transition. Finally, in Section~\ref{sec:discussion} we conclude with a discussion. The details of the $O(N)$ ground state algorithm are provided in Appendix~\ref{app:algorithm}, and Appendix~\ref{app:subparity} contains the derivation of the recursion relation for sub-parities (defined below). Lastly, a first step towards studying the model in the presence of disorder is taken in Appendix~\ref{sec:disorder}, where we show how automatic differentiation may be used to exactly compute thermodynamic quantities for finite system size. We have released the code used for some of the numerical analyses done in this work here: \url{https://github.com/gshartnett/hierarchical}.

\section{The Hierarchical Parity Model \label{sec:model}}
Like all hierarchical models, the hierarchical parity model may be defined recursively. At each step in the recursion, the Hamiltonian of the system is the sum of the left and right sub-system Hamiltonians, plus an interaction term that couples them together. It is convenient to work in terms of the associated balanced binary free, depicted in Fig.~\ref{fig:binary_tree}. The coordinates of the nodes are $k,p$, with $k=0,...,n$ the ``height'' coordinate, measured from the leaf nodes with $k=0$ to the root node with $k=n$, and $p = 1, ..., 2^{n-k}$ the ``width'' coordinate. Using these coordinates, the defining recursion relation for the model is given by:
\begin{equation}
    \label{eq:Hrecursive}
    H_{k,p} := H_{k-1, 2p-1} + H_{k-1, 2p} - J_{k, p} s_{[(p-1) 2^k + 1 : p 2^k]} \,.
\end{equation}
Here the notation $s_{[a:b]}$ denotes the product of spins $a$ through $b$, i.e. $s_{[a:b]} := \prod_{i=a}^b s_i$, and $J_{k,p}$ are arbitrary couplings. The Hamiltonian of the root node corresponds to the full system, $H := H_{n,1}$, and the Hamiltonians of the leaf nodes are simply $H_{0,p} := - J_{0,p} s_p$. At each level in the recursion, two distinct spin-systems are coupled together through their overall parity. The Hamiltonian may also be defined non-recursively as the sum of a parity interaction associated with each node:
\begin{equation}
    \label{eq:modeldefinition}
    H = -\sum_{k=0}^n \sum_{p=1}^{2^{n-k}} J_{k,p} s_{[(p-1) 2^k+1: p \, 2^k]} \,.
\end{equation}
To clarify the notation, the full Hamiltonian for $n=3$ is:
\begin{align}
    H = 
    &- J_{3,1} s_1 s_2 s_3 s_4 s_5 s_6 s_7 s_8 
    - J_{2,1} s_1 s_2 s_3 s_4 
    - J_{2,2} s_5 s_6 s_7 s_8 \nonumber \\
    &- J_{1,1} s_1 s_2 - J_{1,2} s_3 s_4 - J_{1,3} s_5 s_6 - J_{1,4} s_7 s_8 \nonumber \\
    &- J_{0,1} s_1 - J_{0,2} s_2 - J_{0,3} s_3 - J_{0,4} s_4 - J_{0,5} s_5 - J_{0,6} s_6 - J_{0,7} s_7 - J_{0,8} s_8 \,. 
\end{align}

\begin{figure}
    \centering
    \includegraphics[width=0.6\textwidth]{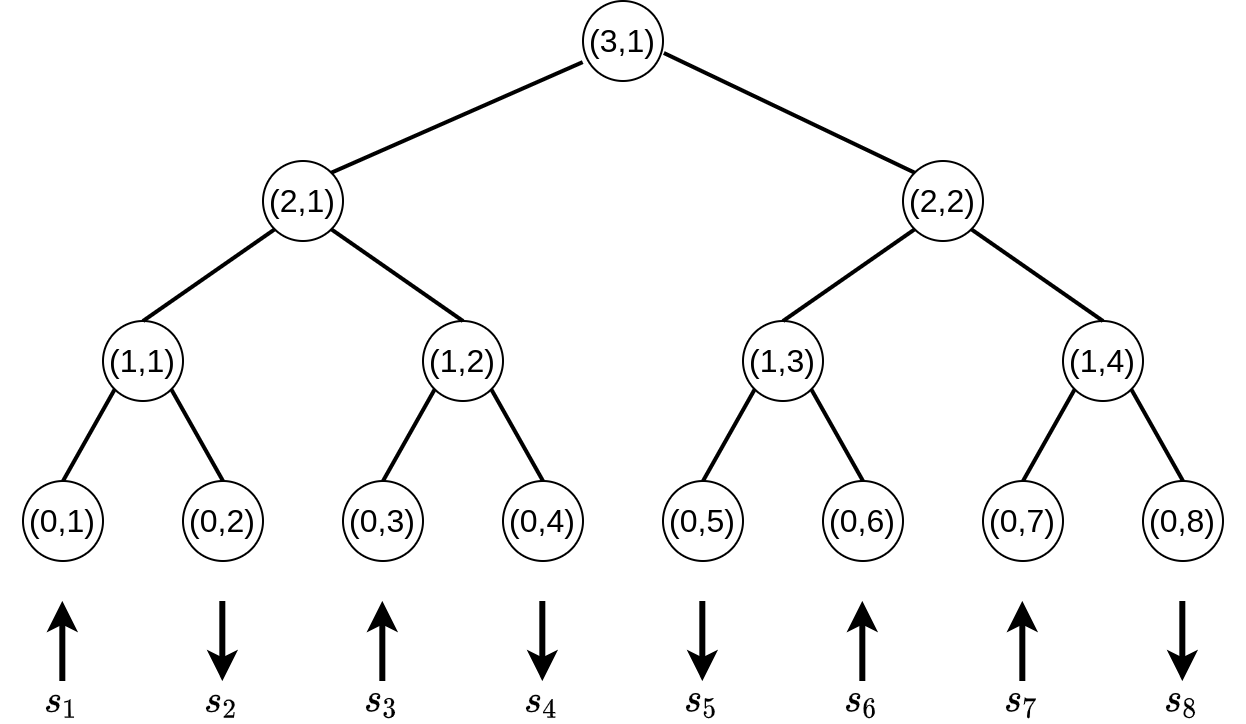}
    \caption{The binary tree corresponding to the the system of $N=2^3=8$ spins, with an arbitrary spin configuration shown below. The nodes are given the coordinates $(k,p)$, with $k$ denoting the height of the tree (measured from the leaf nodes) and $p$ is the width coordinate.}
    \label{fig:binary_tree}
\end{figure}

The recursive property also extends to the partition function. Letting $Z_{k,p}$ denote the partition function associated the $(k,p)$ sub-system with Hamiltonian $H_{k,p}$,
\begin{equation}
    Z_{k,p} := \sum_{\{ s_{(p-1) 2^k + 1}, ..., s_{p 2^k} \}} e^{-\beta H_{k,p}} \,,
\end{equation}
it can be shown that
\begin{align}
    \label{eq:partitionfunction}
    \ln Z_{k,p} &= \ln Z_{k-1, 2p-1} + \ln Z_{k-1, 2p} \\
    &+ \ln \Big[\cosh( \beta J_{k,p}) + P_{k-1,2p-1} P_{k-1, 2p} \sinh(\beta J_{k,p}) \Big] \nonumber \,,
\end{align}
where 
\begin{align}
    P_{k,p} := \beta^{-1} \partial_{J_{k, p}} \ln Z_{k, p} = \langle s_{[(p-1)2^k + 1: p 2^k]} \rangle_{k,p} 
\end{align} 
is the expectation value of the parity of the spins associated with the sub-system $H_{k,p}$.\footnote{Note that the Boltzmann factor appearing in the thermal average is the one associated to the sub-system, $e^{-\beta H_{k,p}}$, and not the Boltzmann factor of the full system, $e^{-\beta H}$.} This expectation value can also be shown to satisfy a recursion relation:
\begin{equation}
    \label{eq:parityrecursion}
    P_{k,p} = \frac{\sinh(\beta J_{k,p}) + \cosh(\beta J_{k,p}) P_{k-1, 2p-1} P_{k-1, 2p}}{\cosh(\beta J_{k,p}) + \sinh(\beta J_{k,p}) P_{k-1, 2p-1} P_{k-1, 2p}} \,.
\end{equation}
The initial conditions of the recursions are that ${Z_{0,p} = 2\cosh(\beta J_{0, p})}$, and ${P_{0,p} = \tanh(\beta J_{0,p})}$.

$P_{k,p}$ is the thermal expectation value of the parity, which can take the values $\pm 1$. Therefore, Eq.~\ref{eq:parityrecursion} may be restated as a recursion relation between the parity probability distribution, rather than the expectation value. Let 
\begin{equation}
    \mathbb{P}_{k,p}(P) := \sum_{\{ s_{(p-1) 2^k + 1}, ..., s_{p 2^k} \}} \frac{e^{-\beta H_{k,p}}}{Z_{k,p}} \delta_{P, s[(p-1)2^k + 1: p 2^k]}
\end{equation}
be the probability mass function over the parity, with $P \in \{-1, 1\}$. This relates to the expectation value via ${P_{k,p} = \mathbb{P}_{k,p}(1) - \mathbb{P}_{k,p}(-1) = 2 \mathbb{P}_{k,p}(1) - 1}$ (since ${\mathbb{P}_{k,p}(-1) = 1 - \mathbb{P}_{k,p}(1)}$). This distribution can be shown to satisfy the recursion: 

\begin{equation}
    \mathbb{P}_{k,p}(P) = \left(\frac{Z_{k-1,2p-1} Z_{k-1,2p-2}}{Z_{k,p}}\right) e^{\beta J_{k,p} P} \sum_{P_L, P_R \in \{-1,1\}} \mathbb{P}_{k-1,2p-1}(P_L) \mathbb{P}_{k-1,2p-2}(P_R) \delta_{P, P_L P_R} \,,
\end{equation}
where $P_{L,R}$ denote the parities of the left and right sub-systems. 

Actually, exact recursion relations can be defined for many more correlators. $P_{k,p}$ is the expectation value of the parity of all $2^k$ spins in the $(k,p)$ sub-system, and it is also interesting to consider the parity of subsets of these $2^k$ spins, also within the $(k,p)$ sub-system. Therefore, introduce the notation $P_{k,p}^{k',p'}$ to denote the parity of the $2^{k'}$ block of spins $s_{[(p'-1)2^{k'}+ 1: p' 2^{k'}]}$ within the $(k,p)$ sub-system, i.e.
\begin{equation}
    P_{k,p}^{k',p'} := \beta^{-1} \partial_{J_{k',p'}} \ln Z_{k,p} \,.
\end{equation}
Note that for this expression to be sensible, $(k', p')$ must be the coordinates of a descendant node of the $(k,p)$ node. To give some examples, for $k'=k-1, p'=2p-1$, this corresponds to the expectation value of the parity of the left descendant spins. For $k'=k-2, p'=4p-3$, this corresponds to the left-left descendant spins, and so on. And $k'=1, p' = 1$ corresponds to $\langle s_1 s_2 \rangle_{k,p}$. 

The recursion relation for $P_{k,p}^{k',p'}$ can be compactly written in terms of a path on the binary tree, at the expense of some additional notation. Let $\bm{n}$ denote the binary tree coordinates, and let $\bm{n}_0 = (k,p)$ denote the node of the sub-system in question, and $\bm{n}_L = (k',p')$ a child node of $\bm{n}_0$ corresponding to the spins of interest. Let $\bm{n}_a$, $a=0,...,L$ denote the sequence of nodes corresponding to the unique length-$L$ path $\mathcal{P}(\bm{n}_0, \bm{n}_L)$ in the tree from the sub-system node $\bm{n}_0$ to the destination node $\bm{n}_L$, with $L = k - k'$. Then, the recursion relation is:
\begin{align}
    \label{eq:subPrecursion}
    P_{\bm{n}_0}^{\bm{n}_L} = \beta^{-1} \partial_{J_{\bm{n}_L}} \sum_{\bm{n} \in \mathcal{P}(\bm{n}_0, \bm{n}_L)} & \ln \Big[ \cosh(\beta J_{\bm{n}}) + P_{\text{Left}(\bm{n})} P_{\text{Right}(\bm{n})} \sinh(\beta J_{\bm{n}}) \Big] \,. 
\end{align}
Here, Left$(\bm{n})$, Right$(\bm{n})$ denote the left or right descendants of node $\bm{n}$. The expression in the sum depends implicitly on $J_{\bm{n}_L}$ via the left or right parity for all nodes in the path except the final one, in which case the dependence is explicit. It is important to note that $P_{k,p}^{k',p'}$ can only represent the parity of collections of spins which comprise all of the descendants of a particular node in the binary tree - therefore, many (but not all) of the correlation functions of the system are governed by exact recursion relations, making this model very amenable to RG analysis. Notably, the nature of the parity interactions allowed these relations to be derived for \textit{arbitrary} couplings.

To gain further physical intuition for the model, a key insight is that it is capable of exhibiting geometric frustration at multiple scales. Considering an arbitrary $H_{k,p}$ sub-system, the top-level interaction encourages the parity of the $2^k$ sub-system spins to be aligned with the sign of the $J_{k,p}$ coupling. Similarly, the descendent interactions encourage the parities of the two $2^{k-1}$ sub-system spins to be aligned with the sign of their respective couplings, $J_{k-1,2p-1}$, $J_{k-1,2p}$. If these are not consistent with one another, i.e. if \begin{equation}
    \label{eq:geometricfrustration}
    \text{sign}(J_{k,p}) \neq \text{sign}(J_{k-1,2p-1} J_{k-1,2p}) \,,
\end{equation} 
then the system will exhibit a form of geometric frustration at the scale of $2^{k-1}$ spins. Moreover, there is a separate frustration condition for each of the $N-1$ $(k,p)$ sub-systems with $k = 1, ..., n$. Therefore, the system can exhibit frustration on multiple scales, provided that at least some of the couplings are anti-ferromagnetic (negative).  

Lastly, an important comment is in order regarding the single-spin couplings, $J_{0,i}$. If these are absent for a system with $N$ spins, then the system can be shown to be equivalent to multiple copies of $N/2$-spin systems. Concretely, if $s_1, ..., s_N$ are the spins of the original system with $J_{0,i} = 0$ $\forall i$, then by introducing the composite spin variables $S_j := s_{2j-1} s_{2j}$, with $j=1,...,N/2$, a new system with half the number of spins is formed which will have single-spin couplings given by the two-body couplings of the original system. For example, the two-body interaction term $J_{1,1} s_1 s_2$ in the original system will become  $J'_{0,1} S_1$ in the new system, with $J'_{0,1} = J_{1,1}$. Moreover, each state in the new system with $N/2$ spins will occur with degeneracy $2^{N/2}$ in the original system with $N$ spins - as required by the fact that there are $2^N$ states in total. Thus, throughout this work we will assume that the single-spin couplings are not all zero, since otherwise the system could be effectively reduced to a simpler system of fewer spins. A consequence of the above is that the single-spin couplings should be thought of as a necessary ingredient in the model definition, and not as an external applied field. This stands in contrast to more standard spin-glass models, such as the Sherrington-Kirkpatrick model or the Edwards-Anderson model.

\begin{figure}
    \centering
    \includegraphics[width=0.48\textwidth]{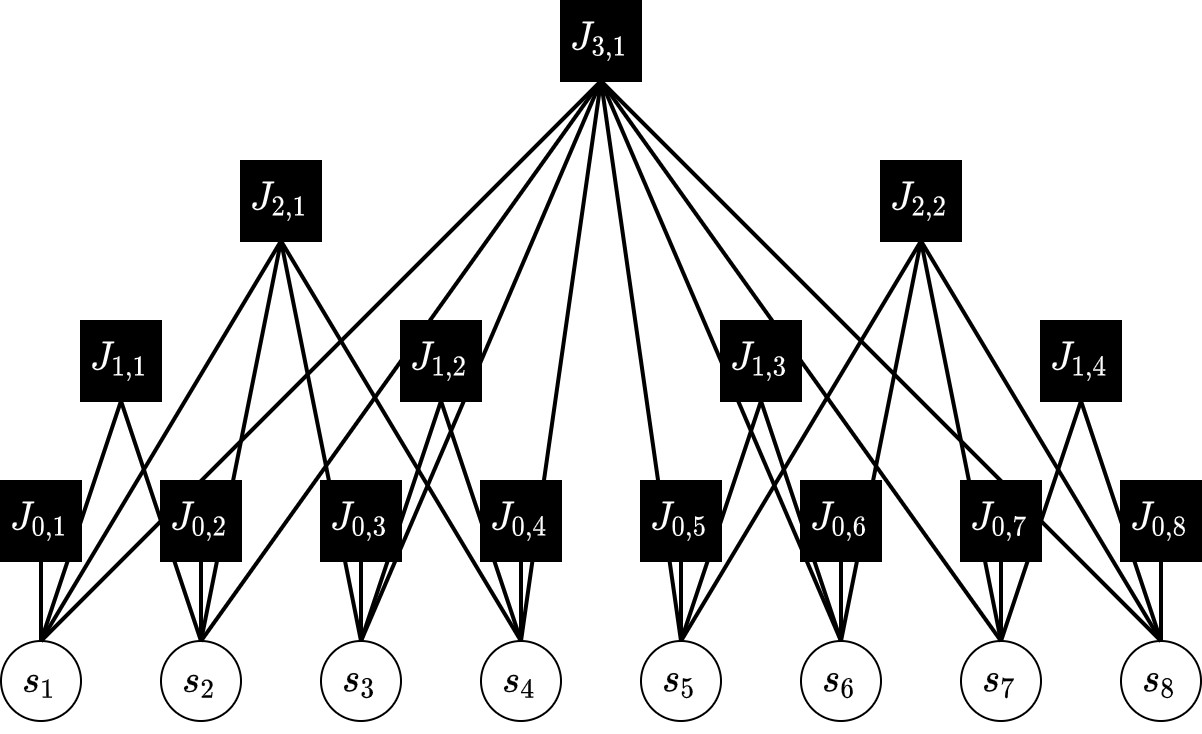}
    \caption{The factor graph for the case $n=3$, corresponding to $N = 2^3 = 8$ spins. The square nodes represent interaction terms, and the circles represent the spin variables.}
    \label{fig:factorgraph}
\end{figure}

\section{Computational Tractability \label{sec:complexity}}
The hierarchical structure is significant enough to render the model solvable, in the sense that many quantities of interest can be computed in $O(N)$ steps.

\begin{prop}
\label{prop:partition}
The partition function is computable in $O(N)$ steps.
\end{prop}

\begin{proof}
By construction. Associate each index pair $(k,p)$ to a node in a balanced binary tree. The recursion relation Eq.~\ref{eq:partitionfunction} allows for the partition function at a given node to be computed, provided the partition function and parity operators of the children nodes are also known. Similarly, Eq.~\ref{eq:parityrecursion} can be used to compute the parity operators of the children nodes in terms of the parities of their children. Each computation takes $O(1)$ steps, and so by carrying out the computation at each of the $2N-1$ nodes in order of increasing $k$ (i.e. from the leaves to the root), the entire partition function may be computed in $O(N)$ steps.
\end{proof}

A similar result holds for the problem of computing the ground state (lowest energy configuration) and its degeneracy.
\begin{prop}
\label{prop:groundstate}
The ground state and its degeneracy is computable and in $O(N)$ steps.
\end{prop}

\begin{proof}
By construction. A general description of the algorithm is provided here; Appendix \ref{app:algorithm} contains a more formal pseudo-code description. Note that if the ground state is degenerate, only a single ground state will be returned by this procedure. The basic idea is that the hierarchical structure of the model allows for the ground state at each node in the binary tree to be computed in terms of the ground states of the left and right children nodes. First, it will be useful to separate the states according to their parity, $\pm 1$, which controls the sign of the interaction between the left and right sub-systems. Let $\bm{s}_{k,p}^{(0)\pm}$ indicate the lowest energy state of the $(k,p)$ sub-system with parity $\pm 1$.
Note that these are the \textit{lowest} energy states, meaning that at least one but possibly both will be a ground state of the full system. In the event that there are multiple lowest energy states at each level, these will just correspond to a single representative state for each parity. 

For a given node $(k,p)$ in the binary tree the lowest energy parity $-1$ state will be one, or both, of $\bm{s}_{k-1,2p-1}^{(0)-} \parallel \bm{s}_{k-1,2p}^{(0)+}$, $\bm{s}_{k-1,2p-1}^{(0)+} \parallel \bm{s}_{k-1,2p}^{(0)-}$. Here $\parallel$ just means the concatenation of the left and right states. Similarly, the lowest energy parity $+1$ state will be one, or both, of $\bm{s}_{k-1,2p-1}^{(0)-} \parallel \bm{s}_{k-1,2p}^{(0)-}$, $\bm{s}_{k-1,2p-1}^{(0)+} \parallel \bm{s}_{k-1,2p}^{(0)+}$. By comparing the energies of these states, the lowest energy state of each parity may be found, together with its degeneracy. Therefore, the $\pm$ lowest-energy state of the full system may be computed recursively by keeping track of the $\pm$ lowest-energy state of each sub-system. The degeneracies may be similarly computed. This procedure is formalized in Algorithm~\ref{alg:groundstate}. Since there are $2N-1$ nodes to visit, and since each node requires $O(1)$ operations, this algorithm runs in $O(N)$ steps. 
\end{proof}

The above results establish that the hierarchical parity model is computationally tractable for arbitrary arbitrary couplings $J_{k,p}$. For general spin systems, computing the partition function or ground state takes an exponential number of steps, and therefore the hierarchical structure of the model allows for exponential speed-ups. 
The notion of a factor graph is often useful in the analysis of the computational tractability of spin systems and graphical models more generally. The factor graph is a bipartite graph with two types of nodes, variable (spin) nodes and factor (interaction) nodes. The factor graph for this model is depicted in Fig.~\ref{fig:factorgraph} for the case of $N = 2^3 = 8$ spins. When the factor graph is a tree (and thus has no loops), belief propagation may be used to efficiently solve a number of computational problems - including the calculation of marginal distributions of a single spin, the sampling of the Boltzmann distribution, and the calculation of the partition function \cite{mezard2009information}. However, as can be seen by direct inspection of Fig.~\ref{fig:factorgraph}, the factor graph in this case is not a tree, and therefore these results do not apply. Thus, Propositions~\ref{prop:partition} and \ref{prop:groundstate} are not merely consequences of well-known results of belief propagation.

\section{Width-Symmetric Model \label{sec:widthsymmetric}}
The above results demonstrate that the model is computational tractable for arbitrary couplings. In this section, we will consider a special case which admits an analytically tractable large-$N$ limit. The width-symmetric model is obtained by restricting the couplings to be independent of the width index $p$, while allowing them to scale with the height index $k$:
\begin{equation}
    J_{k,p} = J_k = 2^{k \sigma} J \,.
\end{equation}
Here $J$ is the bond strength parameter, with $J > 0$ corresponding to ferromagnetic interactions and $J <0$ to anti-ferromagnetic interactions. The scaling of the couplings with height in the binary tree is controlled by $\sigma$.

The nature of the large-$N$ limit is determined by $\sigma$. If $\sigma < 0$, then the parity coupling between sub-systems will vanish in this limit, resulting in an free model. The case $\sigma = 0$ corresponds to the uniform model, where all couplings are equal. For $0 < \sigma < 1$ the couplings scale as a fractional power of the volume (which  is $2^k$ for a level-$k$ sub-system). For other hierarchical models, such as the HEA or HREM, this is known as the non-mean-field regime. The existence of a tractable non-mean-field regime is one of the main motivations for studying hierarchical spin-glasses, as it is more physically relevant to real-world systems than the mean-field regime offered by more widely studied models such as the Sherrington-Kirkpatrick model \cite{sherrington1975solvable}. In contrast, for $\sigma = 1$ the couplings scale linearly with the volume of the system, corresponding to a mean-field system. Finally, the case $\sigma > 1$ corresponds to the case where the couplings scale extra-linearly with system volume. It is important to note that this mean-field terminology does not apply here due to the non-local nature of the interactions. In particular, the existence of a phase transition will be demonstrated for $\sigma = 0$, and the ``non-mean-field" case $0 < \sigma < 1$ will turn out to be trivial. Furthermore, the ``mean-field" case $\sigma = 1$ does not even admit a well-defined thermodynamic limit.

The model may be analyzed for different choices of the scaling parameter $\sigma$: the cases $\sigma < 0$, $\sigma = 0$, and $0 < \sigma < 1$ will be separately considered below. First, however, it is worth noting that as the couplings no longer depend on the ``width'' index $p$, the parity recursion relation Eq.~\ref{eq:parityrecursion} simplifies to
\begin{equation}
    \label{eq:parityrecursion2}
    P_{k} = \frac{\sinh (\beta J_k) + P_{k-1}^2 \cosh (\beta J_k)}{\cosh (\beta J_k) + P_{k-1}^2 \sinh (\beta J_k)} \,,
\end{equation}
where $P_{k,p} = P_k$ for all $p$. The large-$N$ limit may be understood through the analysis of the fixed points of this recursion relation, which are denoted as $P_{\infty} := \lim_{k \rightarrow \infty} P_{k,1}$. Additionally, the relation for the width-symmetric partition function $Z_k = Z_{k,p}$ also simplifies, allowing for the free energy density ${f_n := -\beta^{-1} 2^{-n} \ln Z_{n}}$ to be solved for in terms of a geometrically-weighted sum:
\begin{equation}
    \label{eq:freeenergyrecursion}
    f_n = f_0 - \beta^{-1} \sum_{k=1}^{n} 2^{-k} \ln \left[ \cosh(\beta J_k) + P_{k-1}^2 \sinh(\beta J_k) \right] \,.  
\end{equation}
where ${f_0 = - \beta^{-1} \ln \left[ 2 \cosh(\beta J) \right]}$. 
This series converges, and thus the large-$N$ limit is well-defined, if and only if $\sigma < 1$.

\subsection{Free Model}
We first consider the case $\sigma < 0$. This model can be seen to be free because the coupling between sub-systems vanishes in the large-$N$ limit. The parity expectation value can also be shown to be trivial in this case; the recursion relation becomes $P_{\infty} = P_{\infty}^2$, and so either $P_{\infty} = 0$ or $P_{\infty} = 1$. Expanding around these reveals that the former fixed point is stable, whereas the latter is unstable.

\subsection{Uniform Model}
Setting $\sigma = 0$ corresponds to a uniform model, where the couplings are all the same. As illustrated in Fig.~\ref{fig:factorgraph}, the interaction terms in the Hamiltonian can be associated with the nodes in a balanced binary tree (the black squares), and as such there are $2N-1$ such terms. For $\sigma = 0$, each term contributes equally to the overall Hamiltonian, which means that $N$ interactions involve a single spin, $N/2$ spins involve 2 spins, $N/4$ interactions involve 4 spins, and so on, with just a single term involving the full system parity.

\subsubsection{The Parity Recursion Relation}
The fixed points of the parity recursion relation Eq.~\ref{eq:parityrecursion2} are given by the roots of a cubic polynomial, and are: $P_{\infty}^{(0)} = 1$ and
\begin{equation}
    P_{\infty}^{(\pm)} = \frac{\text{coth}(\beta J) -1}{2} \pm \frac{\text{sgn}(\beta J)}{2} \sqrt{ (\text{coth}(\beta J) - 3) (\text{coth}(\beta J) + 1)} \,.
\end{equation}
The 3 roots are degenerate for $\beta = \beta_c$, with
\begin{equation}
    \beta_c = \frac{\text{coth}^{-1}(3)}{J} = \frac{\ln 2}{2 J} \approx \frac{0.3466}{J} \,.
\end{equation}

It is convenient to consider the ferromagnetic and anti-ferromagnetic cases separately. For the ferromagnetic case with $J > 0$, the $P^{(\pm)}_{\infty}$ roots are complex for $\beta > \beta_c$, and they are real for $\beta < \beta_c$. However in this case $P_{\infty}^{(+)}$ is greater than 1 and thus represents an unphysical solution for all temperatures, since the parity expectation value must lie in the interval $[-1,1]$. Moreover, the $P^{(0)}$ fixed-point is stable for $\beta > \beta_c$ and unstable for $\beta < \beta_c$, whereas the $P^{(\pm)}$ fixed-points are stable for $\beta < \beta_c$ and unstable for $\beta > \beta_c$. We therefore conclude that 
\begin{equation}
    P_{\infty} = 
    \begin{cases} 
    P_{\infty}^{(-)} &\mbox{if } \beta \le \beta_c \\
    P_{\infty}^{(0)} & \mbox{if } \beta \ge \beta_c 
    \end{cases}
    \,.
\end{equation}
Just above the critical temperature the correlation function behaves as
\begin{equation}
    P_{\infty}^{(-)} = 1 - 2 \sqrt{2(\beta_c - \beta) J} + 4 (\beta_c - \beta) J + \mathcal{O}\left( (\beta_c - \beta)^{3/2} \right) \,,
\end{equation}
and therefore the parity expectation value is non-analytic around the critical temperature $\beta = \beta_c$, indicating a phase transition. 

In contrast, for the anti-ferromagnetic case with ${J < 0}$, all 3 fixed-points are real for all temperatures - although $P_{\infty}^{(+)} < -1$ and once again represents an unphysical solution. The $P^{(0)}$ fixed-point is always unstable, and the $P^{(\pm)}$ fixed-points are always stable. Therefore, the physically realized solution is $P_{\infty} = P_{\infty}^{(-)}$ and there is no non-analytic behavior. In order to understand the convergence of $P_{n,1}$ to the fixed point $P_{\infty}$, in Fig.~\ref{fig:Pkplot_Jboth} $P_{n,1}$ is depicted for $n=0,1,2,...,10$ (corresponding to $N=1, 2, 4, ..., 1024$) for both the anti-ferromagnetic and ferromagnetic cases. 

\begin{figure}
    \centering
    \includegraphics[width=0.6\textwidth]{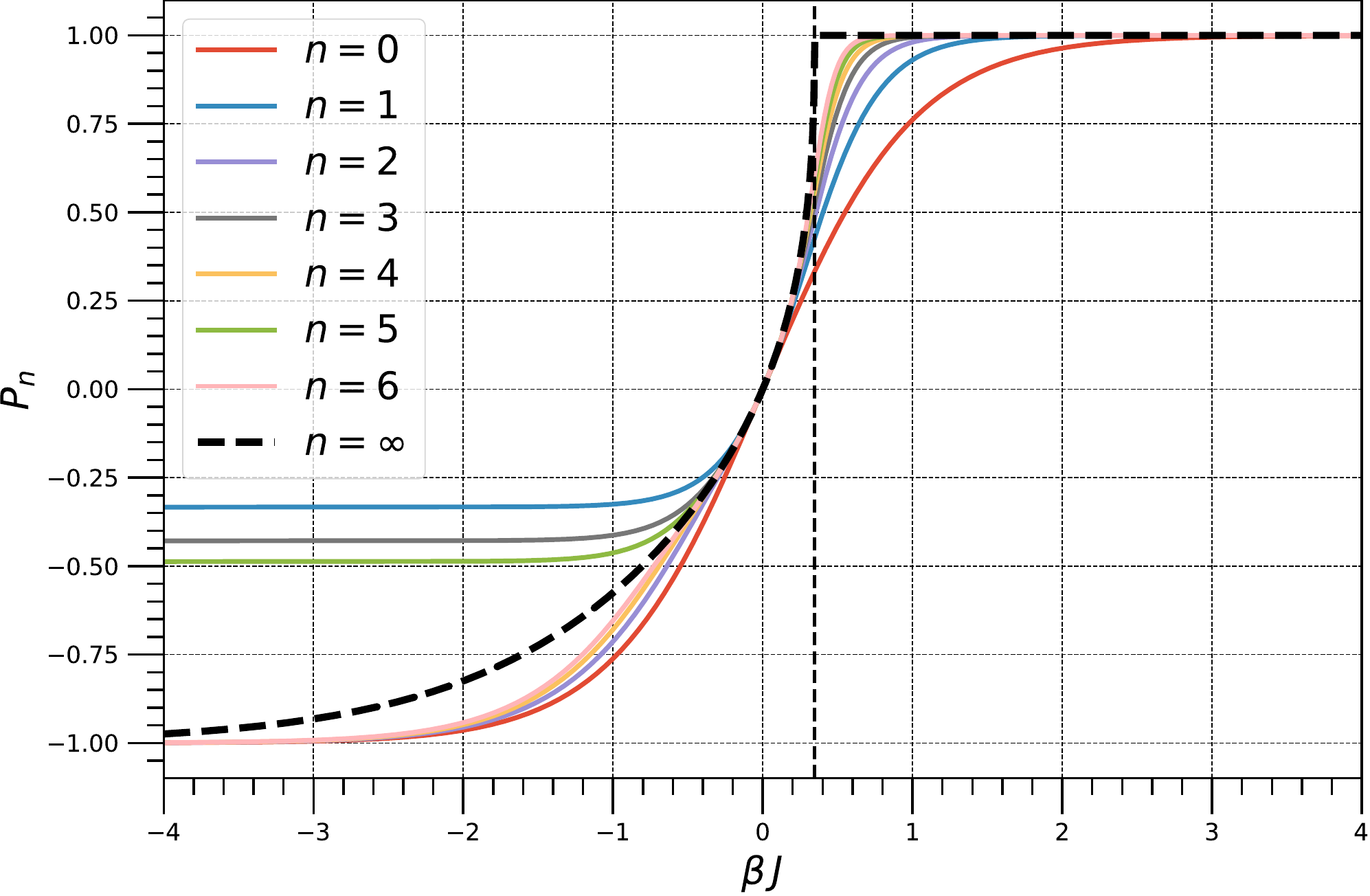}
    \caption{The parity expectation value $P_{n}$ as a function of inverse temperature in the uniform model, for both the anti-ferromagnetic $(J < 0)$ and ferromagnetic $(J>0)$ models. To emphasize the relation between the two cases, the anti-ferromagnetic curve has been plotted for negative $\beta J$ values and the ferromagnetic curve has been plotted for positive $\beta J$ values. The thick dashed line shows the analytic result for $P_{\infty}$, and the vertical dashed line marks the phase transition.}
    \label{fig:Pkplot_Jboth}
\end{figure}

To summarize, the uniform $(\sigma = 0$) width-symmetric model exhibits a phase transition in the ferromagnetic case where $J > 0$. The order parameter is the expectation value of the parity of the full system, ${P_{\infty} = \lim_{n \rightarrow \infty} \langle s_1 ... s_{2^n} \rangle}$. This quantity is 1 in the low-temperature parity-locked phase, and lies in the interval $(0,1)$ in the high-temperature phase.\footnote{Technically, it is more accurate to call $1 - P_{\infty}$ the order parameter. However, this terminology seems counter-intuitive since it suggests that the parity-locked phase with $P_{\infty} = 1$ is ``disordered".}  Importantly, there is no geometric frustration in this case as all the couplings are positive. In contrast, there is no phase transition in the anti-ferromagnetic model, and the parity is not locked for any finite temperature, although it does tends toward $-1$ as $\beta \rightarrow \infty$. The anti-ferromagnetic model does exhibit geometric frustration, leading to an extensive ground state degeneracy (see below). Unfortunately, although this model allows for many relations to be worked out, it does not seem possible to obtain an expression for the free energy density. This is because the sum in Eq.~\ref{eq:freeenergyrecursion} depends on $P_k$ for all $k$, and not just on the fixed points.

The nature of the phase transition can be further probed by studying the fixed points of the sub-parities via Eq.~\ref{eq:subPrecursion}. Let $P_{\infty}^{(L)} := \lim_{n \rightarrow \infty} P_{n,1}^{k,p}$ denote the large-$N$ limit of the sub-parities in the uniform coupling case, with $L = n - k$ denoting the length of the path from the root node $(n,1)$ to the destination node $(k,p)$. As shown in the SI, the ferromagnetic parity-locked phase exhibits $P_{\infty}^{(L)} = P_{\infty} = 1$. All the sub-parities are 1, which means that every single spin is pointing up. Additionally, just above the transition temperature the parities obey $P_{\infty}^{(L)} \sim 1 - c_L \sqrt{(\beta_c - \beta)J}$, for $c_L$ some positive constant. Thus, as the temperature is raised above the critical temperature the sub-parities become un-frozen at every scale.

\subsubsection{The Ground State Degeneracy}
In Algorithm~\ref{app:algorithm} an $O(N)$ algorithm for computing the ground state, energy, and degeneracy for a system of arbitrary couplings is presented. The uniform model represents a special case of the more general problem for which the algorithm may be simplified to a set of recursion relations, which we now derive. First, we note that the ferromagnetic model has a unique ground state corresponding to the all up configuration $s_i = 1$ $\forall i$. To calculate the ground state degeneracy of the anti-ferromagnetic model, first consider the system with $n=1$, corresponding to $N=2$ spins. The configurations $\{\downarrow\downarrow, \uparrow\downarrow, \downarrow\uparrow\}$ each have energy $-1$, while the configuration $\uparrow\uparrow$ has energy $3$ (setting $J=-1$ for convenience).\footnote{For brevity here we use $\downarrow$ to correspond to the $s=-1$ spin down state and $\uparrow$ to correspond to the $s=1$ spin up state.} So, the ground state degeneracy for level $n=1$ is $d_1 = 3$. Also, note that the first ground state, $\downarrow\downarrow$, has even parity, whereas the other two ground states have odd parity. Candidate ground states at level $n=2$ can be formed by concatenating two copies of these states. There are five even parity configurations, four are formed by joining two odd parity configurations, such as $(\uparrow\downarrow)(\downarrow\uparrow)$, and one is formed by two copies of the sole even parity configuration, $(\downarrow\downarrow)(\downarrow\downarrow)$. Each of these has energy $-1$. There are also four odd parity configurations of the form $(\downarrow\downarrow)(\uparrow \downarrow)$. These have energy $-3$, and so the ground state degeneracy at level $n=2$ is $d_2 = 4$. 

In considering the ground state degeneracy at level $n=3$, even parity configurations can be formed by joining the ground states at level $n=2$, these will have energy $-3-3+1 = -5$ (two copies of a $-3$ energy state plus a parity penalty term of $+1$). At first glance this might seem to exhaust the ground states at this level, but because the energy cost of flipping the overall parity happens to be the same as the gap between the ground and first excited states of the previous level system, odd parity ground states at level $n=3$ can be formed by joining an odd parity $n=2$ ground state with an even parity $n=2$ first excited state, for example $(\downarrow\downarrow\uparrow\downarrow)(\downarrow\downarrow\downarrow\downarrow)$. These states will also have energy $-3-1-1 = -5$. The ground state degeneracy at level $n=3$ is therefore $d_3 = 56$, consisting of 16 even parity states plus 40 odd parity states. A simple set of recursion relations that captures this pattern is:
\begin{equation}
    d_{n}^- = 2 \, d_{n-1}^- \, d_{n-1}^+ \,, 
    \qquad         
    d_n^+ = 
    \begin{cases} (d_{n-1}^-)^2 & n \text{ odd} \\ 
        (d_{n-1}^-)^2 + (d_{n-1}^+)^2 & n \text{ even} 
    \end{cases}  \,, 
    \qquad         
    d_n = 
    \begin{cases} d_n^- + d_n^+ & n \text{ odd} \\ 
        d_n^- & n \text{ even}
    \end{cases} \,.
\end{equation}
Here $d_n$ is the ground state degeneracy and $d_n^{\pm}$ are the degeneracies of the lowest odd/even parity states, respectively, which may or may not be themselves ground states. The recursion begins with all terms equal to 1 for $n=0$. The large-$N$ behavior of the degeneracy $d_n$ can be numerically inferred to be 
\begin{equation}
    d_n \sim A \, b^N \,,
\end{equation}
with $A$ an overall constant, and $b \approx 1.6234$. This corresponds to an extensive ground state degeneracy. 

\subsection{Parity-Locked Model}
Lastly, for $(0 < \sigma < 1)$ the fixed point equation, Eq.~\ref{eq:parityrecursion}, becomes:
\begin{equation}
    P_{\infty} = \text{sgn}(J) \,.
\end{equation}
As a result, the model is always in the ``parity-locked phase'', and there is no phase transition.

\section{Discussion \label{sec:discussion}}
In this work we have introduced a hierarchical parity model. Like all hierarchical models, the recursive property greatly facilitates an RG analysis, and in particular exact RG recursion relations may be derived for the parity of full system, as well as for certain subsets of spins. Moreover, these relations hold in the presence of \textit{arbitrary} couplings. The model also admits $O(N)$ algorithms for computing the ground state and the partition function. There are other well-known computationally tractable examples of spin models with arbitrary couplings. For example, one important result is that 2-spin Ising models defined over planar lattices are solvable in polynomial time \cite{barahona1982computational}. This result can be extended to graphs of fixed genus $g$ (planar graphs have genus $g=0$) \cite{regge2000combinatorial, galluccio2000new}. Importantly, these results are for general couplings. If the couplings are homogeneous and ferromagnetic, then the partition function of the Ising model defined over a general graph $G$ can be recognized as the Tutte polynomial $T_G(x,y)$, evaluated at a particular point in $\mathbb{R}^2$ \cite{welsh1990computational}. The computational complexity of computing this is {\#P-hard} in general, except for a few special isolated points as well as all points lying on the hyperbola $H_1: (x-1)(y-1) = 1$, in which case it is computable in polynomial time. Notably, all of these examples involve models with 2-spin interaction terms, whereas Eq.~\ref{eq:modeldefinition} is characterized by interaction terms involving $2^k$-spin terms, with $k=0,...,n$. This work thus expands the set of Ising spin system geometries known to remain tractable even in the presence of frustration.


In order to study the thermodynamic, large-$N$ limit of the model it is convenient to restrict the couplings to only depend on the height of the binary tree. For $J_{k,p} = 2^{k\sigma} J$, the model is trivially free for $\sigma < 0$, and ill-defined for $\sigma \ge 1$. The case $\sigma = 0$ corresponds to a uniform model where all the couplings are equal, and this model was shown to exhibit a thermal phase transition. The low-temperature phase is ``parity-locked''. Just above the transition, the parity correlators exhibit non-analyticity as they melt and become unfrozen. Finally, the case $0 < \sigma < 1$, which would correspond to a non-mean-field model in a hierarchical model with local interactions, turns out to be trivial in that the parities are always locked and there is no evidence of a phase transition.

Lastly, one of the key motivations for studying hierarchical models such as this one is to better understand the Renormalization Group for complex disordered systems. As a first step towards this goal, in Appendix~\ref{sec:disorder} we utilize the recursion relation for the partition function and the ground state algorithm to compute key thermodynamic quantities of the model at finite system size. However, more work is needed to properly investigate whether the disordered model exhibits a spin-glass transition.

\subsection*{Data Availability}
Data sharing not applicable to this article as no datasets were generated or analyzed during the current study.

\subsection*{Acknowledgments}
This work grew out of an earlier collaboration with Masoud Mohseni, and I wish to acknowledge useful discussions with him throughout this project. I would also like to thank Edward Parker, Federico Ricci-Tersenghi, and two anonymous referees for their useful feedback on earlier versions of this manuscript.

\appendix

\section{Ground State Algorithm \label{app:algorithm}}

\begin{algorithm}[H]
  \caption{Compute the ground state, energy, and degeneracy for each parity.}
  \label{alg:groundstate}
  \begin{algorithmic}[1]
    \Function{GroundState}{$k, p$}
      \If{$k=0$}
        \State $\bm{s}_{0,p}^{(0)-} = -1$, $E_{0,p}^{(0)-} = J_{0,p}$, $d_{0,p}^{(0)-} = 1$,
        \State $\bm{s}_{0,p}^{(0)+} = 1$, $E_{0,p}^{(0)+} = -J_{0,p}$, $d_{0,p}^{(0)+} = 1$,        
      \Else
        \State $\bm{s}_{k-1,2p-1}^{(0)-}, E_{k-1,2p-1}^{(0)-}, d_{k-1,2p-1}^{(0)-}, \bm{s}_{k-1,2p-1}^{(0)+}, E_{k-1,2p-1}^{(0)+}, d_{k-1,2p-1}^{(0)+} = \text{GroundState}(k-1, 2p-1)$
        \State $\bm{s}_{k-1,2p}^{(0)-}, E_{k-1,2p}^{(0)-}, d_{k-1,2p}^{(0)-}, \bm{s}_{k-1,2p}^{(0)+}, E_{k-1,2p}^{(0)+}, d_{k-1,2p}^{(0)+} = \text{GroundState}(k-1, 2p)$    
        \If{$E_{k-1,2p-1}^{(0)-} + E_{k-1,2p}^{(0)+} <  E_{k-1,2p-1}^{(0)+} + E_{k-1,2p}^{(0)-}$} \Comment{Compute the $-$ parity quantities}
            \State $\bm{s}_{k,p}^{(0)-} = \bm{s}_{k-1,2p-1}^{(0)-} \parallel \bm{s}_{k-1,2p}^{(0)+}$
            \State $E_{k,p}^{(0)-} = E_{k-1,2p-1}^{(0)-} + E_{k-1,2p}^{(0)+} + J_{k,p}$
            \State $d_{k,p}^- = d_{k-1,2p-1}^- \, d_{k-1,2p}^+$
        \ElsIf{$E_{k-1,2p-1}^{(0)-} + E_{k-1,2p}^{(0)+} >  E_{k-1,2p-1}^{(0)+} + E_{k-1,2p}^{(0)-}$}
             \State $\bm{s}_{k,p}^{(0)-} = \bm{s}_{k-1,2p-1}^{(0)+} \parallel \bm{s}_{k-1,2p}^{(0)-}$
            \State $E_{k,p}^{(0)-} = E_{k-1,2p-1}^{(0)+} + E_{k-1,2p}^{(0)-} + J_{k,p}$
            \State $d_{k,p}^- = d_{k-1,2p-1}^+ \, d_{k-1,2p}^-$
        \Else
            \State $\bm{s}_{k,p}^{(0)-} = \bm{s}_{k-1,2p-1}^{(0)-} \parallel \bm{s}_{k-1,2p}^{(0)+}$
            \State $E_{k,p}^{(0)-} = E_{k-1,2p-1}^{(0)-} + E_{k-1,2p}^{(0)+} + J_{k,p}$
            \State $d_{k,p}^- = d_{k-1,2p-1}^- \, d_{k-1,2p}^+ + d_{k-1,2p-1}^+ \, d_{k-1,2p}^-$
        \EndIf
        \If{$E_{k-1,2p-1}^{(0)-} + E_{k-1,2p}^{(0)-} <  E_{k-1,2p-1}^{(0)+} + E_{k-1,2p}^{(0)+}$} \Comment{Compute the $+$ parity quantities}
            \State $\bm{s}_{k,p}^{(0)+} = \bm{s}_{k-1,2p-1}^{(0)-} \parallel \bm{s}_{k-1,2p}^{(0)-}$
            \State $E_{k,p}^{(0)+} = E_{k-1,2p-1}^{(0)-} + E_{k-1,2p}^{(0)-} - J_{k,p}$
            \State $d_{k,p}^+ = d_{k-1,2p-1}^- \, d_{k-1,2p}^-$
        \ElsIf{$E_{k-1,2p-1}^{(0)-} + E_{k-1,2p}^{(0)-} >  E_{k-1,2p-1}^{(0)+} + E_{k-1,2p}^{(0)+}$}
             \State $\bm{s}_{k,p}^{(0)+} = \bm{s}_{k-1,2p-1}^{(0)+} \parallel \bm{s}_{k-1,2p}^{(0)+}$
            \State $E_{k,p}^{(0)+} = E_{k-1,2p-1}^{(0)+} + E_{k-1,2p}^{(0)+} - J_{k,p}$
            \State $d_{k,p}^+ = d_{k-1,2p-1}^+ \, d_{k-1,2p}^+$
        \Else
            \State $\bm{s}_{k,p}^{(0)+} = \bm{s}_{k-1,2p-1}^{(0)+} \parallel \bm{s}_{k-1,2p}^{(0)+}$
            \State $E_{k,p}^{(0)+} = E_{k-1,2p-1}^{(0)+} + E_{k-1,2p}^{(0)+} - J_{k,p}$
            \State $d_{k,p}^+ = d_{k-1,2p-1}^- \, d_{k-1,2p}^- + d_{k-1,2p-1}^+ \, d_{k-1,2p}^+$
        \EndIf 
      \EndIf
        \State \Return $\bm{s}_{k,p}^{(0)-}, E_{k,p}^{(0)-}, d_{k,p}^{(0)-}, \bm{s}_{k,p}^{(0)+}, E_{k,p}^{(0)+}, d_{k,p}^{(0)+}$ 
    \EndFunction
  \end{algorithmic}
\end{algorithm}

\section{Sub-Parity Recursion Relation and Fixed-Point Analysis \label{app:subparity}} 
Recall that an exact recursion relation may be derived for the ``sub-parity'' correlation function $P_{k,p}^{k',p'}$ (Eq.~\ref{eq:subPrecursion}):
\begin{align}
    P_{\bm{n}_0}^{\bm{n}_L} = \beta^{-1} \partial_{J_{\bm{n}_L}} \sum_{\bm{n} \in \mathcal{P}(\bm{n}_0, \bm{n}_L)} & \ln \Big[ \cosh(\beta J_{\bm{n}}) + P_{\text{Left}(\bm{n})} P_{\text{Right}(\bm{n})} \sinh(\beta J_{\bm{n}}) \Big] \,. 
\end{align}
Here $\bm{n}$ is used to denote the coordinates of a node in the binary tree, with $\bm{n}_0 = (k,p)$ and $\bm{n}_L = (k',p')$, and where $\mathcal{P}(\bm{n}_0, \bm{n}_L)$ is the unique path from $\bm{n}_0$ to $\bm{n}_L$. The length of the path is given by $L = k - k'$. Additionally, Left$(\bm{n})$, Right$(\bm{n})$ denote the left or right descendants of node $\bm{n}$. Fixed-point equations for these operators may be defined. Without loss of generality, assume that the path $\mathcal{P}_{\bm{n}_0, \bm{n}_L}$ is totally left, in other words $\bm{n}_L = \text{Left}^L(\bm{n}_0)$. 
First, note that the the terms in the sum for all $\bm{n} \neq \bm{n}_L$ depend on $J_{\bm{n}_L}$ implicitly through the left parities, and never the right parities, because the path is totally left. In contrast, the term in the sum corresponding to the destination node $\bm{n}_L$ depends explicitly on $J_{\bm{n}}$ through the hyperbolic cosine and sine terms. Therefore:
\begin{align}
    P_{\bm{n}_0}^{\bm{n}_L} = \beta^{-1} \sum_{\bm{n} \in \mathcal{P}(\bm{n}_0, \bm{n}_{L-1})} & \frac{\left( \frac{\partial P_{\text{Left}(\bm{n})}}{\partial J_{\bm{n}_L}}\right) P_{\text{Right}(\bm{n})} \sinh(\beta J_{\bm{n}}) }{\cosh(\beta J_{\bm{n}}) + P_{\text{Left}(\bm{n})} P_{\text{Right}(\bm{n})} \sinh(\beta J_{\bm{n}})} \\
    &+ \Bigg[ \frac{\sinh(\beta J_{\bm{n}_L}) +  P_{\text{Left}(\bm{n}_L)} P_{\text{Right}(\bm{n}_L)} \cosh(\beta J_{\bm{n}_L})}{\cosh(\beta J_{\bm{n}_L}) + P_{\text{Left}(\bm{n}_L)} P_{\text{Right}(\bm{n}_L)} \sinh(\beta J_{\bm{n}_L})} \Bigg] \nonumber \,. 
\end{align}
The second term can be recognized as being simply $P_{\bm{n}_L}$, the parity of the sub-system associated with the node $\bm{n}_L$. To evaluate the first term, use the chain rule:
\begin{equation}
    \frac{\partial P_{\text{Left}(\bm{n}_{\ell})}}{\partial J_{\bm{n}_L}} = 
    \frac{\partial P_{\text{Left}(\bm{n}_L)}}{\partial J_{\bm{n}_L}} 
    \frac{\partial P_{\text{Left}(\bm{n}_{L-1})}}{\partial P_{\text{Left}(\bm{n}_{L})}}
    \frac{\partial P_{\text{Left}(\bm{n}_{L-2})}}{\partial P_{\text{Left}(\bm{n}_{L-1})}}
    \ldots
    \frac{\partial P_{\text{Left}(\bm{n}_{\ell})}}{\partial P_{\text{Left}(\bm{n}_{\ell+1})}} \,,
\end{equation}
where $\ell=0,...,L-1$. Finally, from Eq.~\ref{eq:parityrecursion} it can be worked out that
\begin{equation}
    \beta^{-1} \frac{\partial P_{\bm{n}}}{\partial J_{\bm{n}}} = \frac{\left(1 - P_{\text{Left}(\bm{n})}^2 P_{\text{Right}(\bm{n})}^2 \right) }{\left( \cosh(\beta J_{\bm{n}}) + P_{\text{Left}(\bm{n})} P_{\text{Right}(\bm{n})} \sinh(\beta J_{\bm{n}}) \right)^2} \,,
\end{equation}
and that
\begin{equation}
    \frac{\partial P_{\bm{n}}}{\partial P_{\text{Left}(\bm{n})}} = \frac{P_{\text{Right}(\bm{n})}}{\left( \cosh(\beta J_{\bm{n}}) + P_{\text{Left}(\bm{n})} P_{\text{Right}(\bm{n})} \sinh(\beta J_{\bm{n}}) \right)^2} \,.
\end{equation}

Next, take all the couplings to be equal and take all the parities $P_{\bm{n}}$ to be fixed points, i.e. $P_{\bm{n}} = P_{\infty}$. Then,
\begin{equation}
    P_{\infty}^{(L)} = P_{\infty}\left[ 1 + \frac{(1 - P_{\infty}^4) \sinh(\beta J)}{\left(\cosh(\beta J) + P_{\infty}^2 \sinh(\beta J) \right)^3}  \sum_{\ell=0}^{L-1} \left( \frac{P_{\infty}}{\left( \cosh(\beta J) + P_{\infty}^2 \sinh(\beta J) \right)^2} \right)^{L-\ell} \right] \,,
\end{equation}
where the first factor has been simplified using the fixed-point equation. The series is geometric, and can be summed for general $L$, although the resulting expression is not particularly illuminating. However, it is easy to see that $P_{\infty}^{(L)} = 1$ for $P_{\infty} = 1$. Furthermore, near the phase transition ${P_{\infty} \sim 1 - c \sqrt{(\beta_c - \beta)J}}$, with $c = 2 \sqrt{2}$, and it can also be shown that ${P_{\infty}^{(L)} \sim 1 - c_L \sqrt{(\beta_c - \beta) J}}$, with $c_L = \sqrt{2} (1 + 2^{-L})$. Thus, all the sub-parities exhibit the same non-analyticity as the phase transition is crossed.

\section{Disordered Model \label{sec:disorder}}
The model may also be studied in the presence of disorder. The motivation for doing so is that the recursive structure of the model could allow for greater insight into the Renormalization Group and how it relates to replica symmetry breaking for spin-glasses, as was done for example in \cite{castellana2010hierarchical, castellana2010renormalization, castellana2011real, decelle2014ensemble}. Unfortunately, but unsurprisingly, the presence of disorder greatly complicates the analytic treatment of the recursion relations. Therefore, in this Appendix we will simply show how the recursions may be numerically solved to compute certain thermodynamic quantities, and we leave a more thorough analysis of the disordered model for future work. The partition may be computed exactly for a given disorder realization and for finite system size. For concreteness, we will consider a $\pm J$ disordered model, where the couplings are distributed according to
\begin{equation}
    J_{k,p} = 
    \begin{cases} 
    + J & \text{ with probability } p \\
    - J & \text{ with probability } 1-p \,.
    \end{cases}    
\end{equation}


We used the automatic differentiation capabilities of the Python software package Jax \cite{jax2018github} to exactly compute the following thermodynamic quantities for finite $N$: the free energy ${F = - \ln Z/\beta}$, the energy ${E = - \partial_{\beta} \ln Z}$, the heat capacity ${C_V = \beta^2 \partial_{\beta}^2 \ln Z}$, and the entropy ${S = \beta (E - F)}$. Fig.~\ref{fig:disordered_model_thermodynamics} depicts contour plots of these quantities as functions of inverse temperature and disorder probability for a system size of $N=256$ spins ($n=9$).\footnote{The code for this analysis has been made available here: \url{https://github.com/gshartnett/hierarchical}.}

\begin{figure*}
    \centering
    \includegraphics[width=1.0\textwidth]{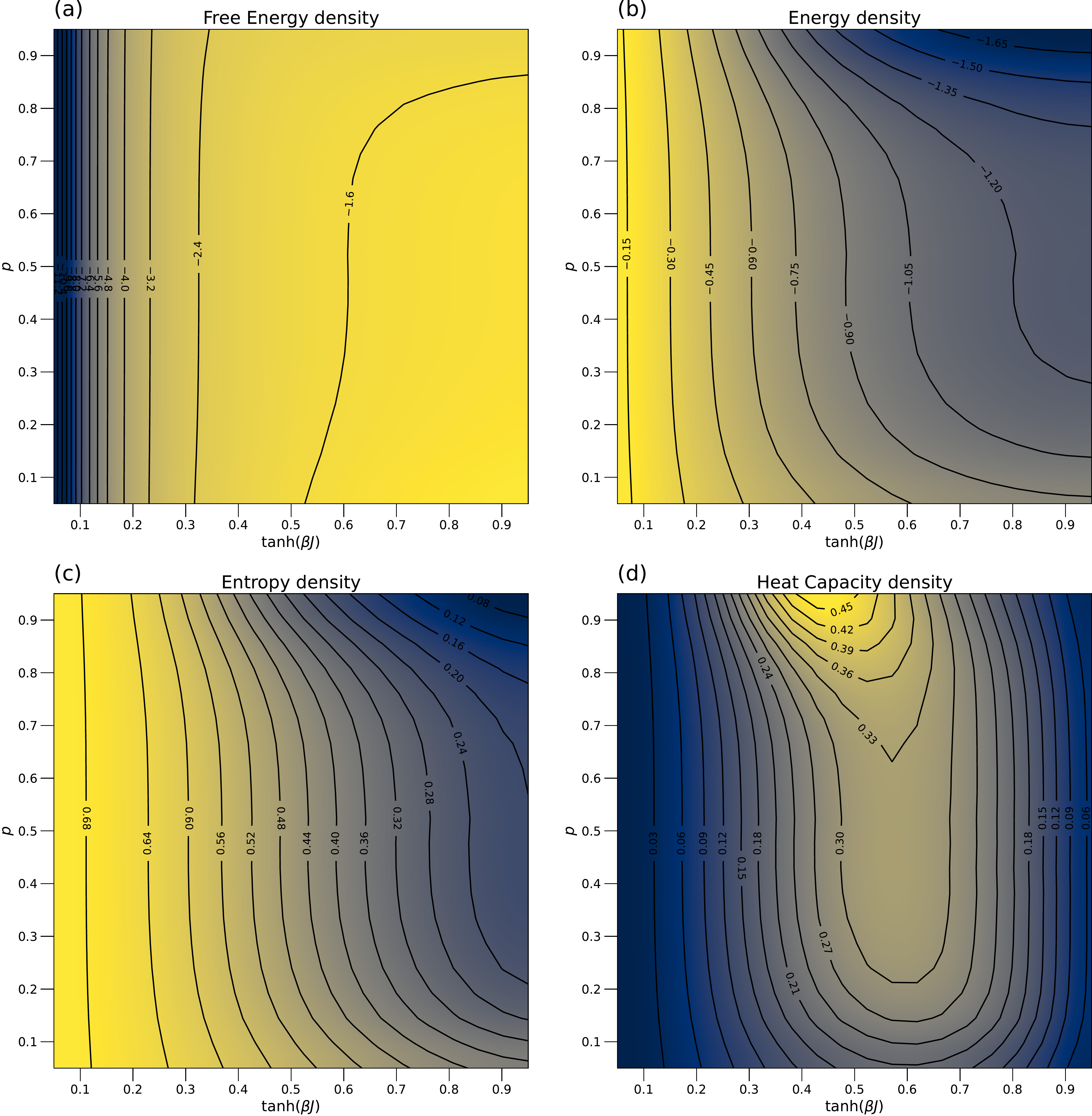}
    \caption{Contour plots depicting the (a) free energy density, (b) energy density, (c) entropy density, and (d) heat capacity density as a function of $\tanh(\beta J)$ and disorder probability $p$ for a system size of $N=512$ spins ($n=9$) and 100 disorder averages.}
    \label{fig:disordered_model_thermodynamics}
\end{figure*}

Automatic differentiation may also be used to compute derivatives with respect to the individual couplings, $J_{k,p}$. In particular, the spin-glass susceptibility matrix is defined as:
\begin{equation}
    \chi_{p p'} = \langle s_p s_{p'} \rangle - \langle s_p \rangle \langle s_{p'} \rangle = \frac{1}{\beta^2} \frac{\partial^2 \ln Z}{\partial J_{0,p} \partial J_{0,p'}} \,,
\end{equation}
where $p, p' =1, ..., 2^{n}$ index the spins. (Recall that $J_{0,p}$ is just the local external field at spin $p$.) From this matrix the spin-glass susceptibility, defined as $\chi_{SG} = \frac{\beta^2}{N} \sum_{p,p'} \chi_{ij}^2$, may be easily computed. Fig.~\ref{fig:disordered_model_chi} depicts a contour plot for $\chi_{SG}$ in the $(\beta, p)$ plane and a histogram of the eigenvalues of the $\chi_{pp'}$ matrix for the case $p=1/2$. 

\begin{figure*}
    \centering
    \includegraphics[width=1.0\textwidth]{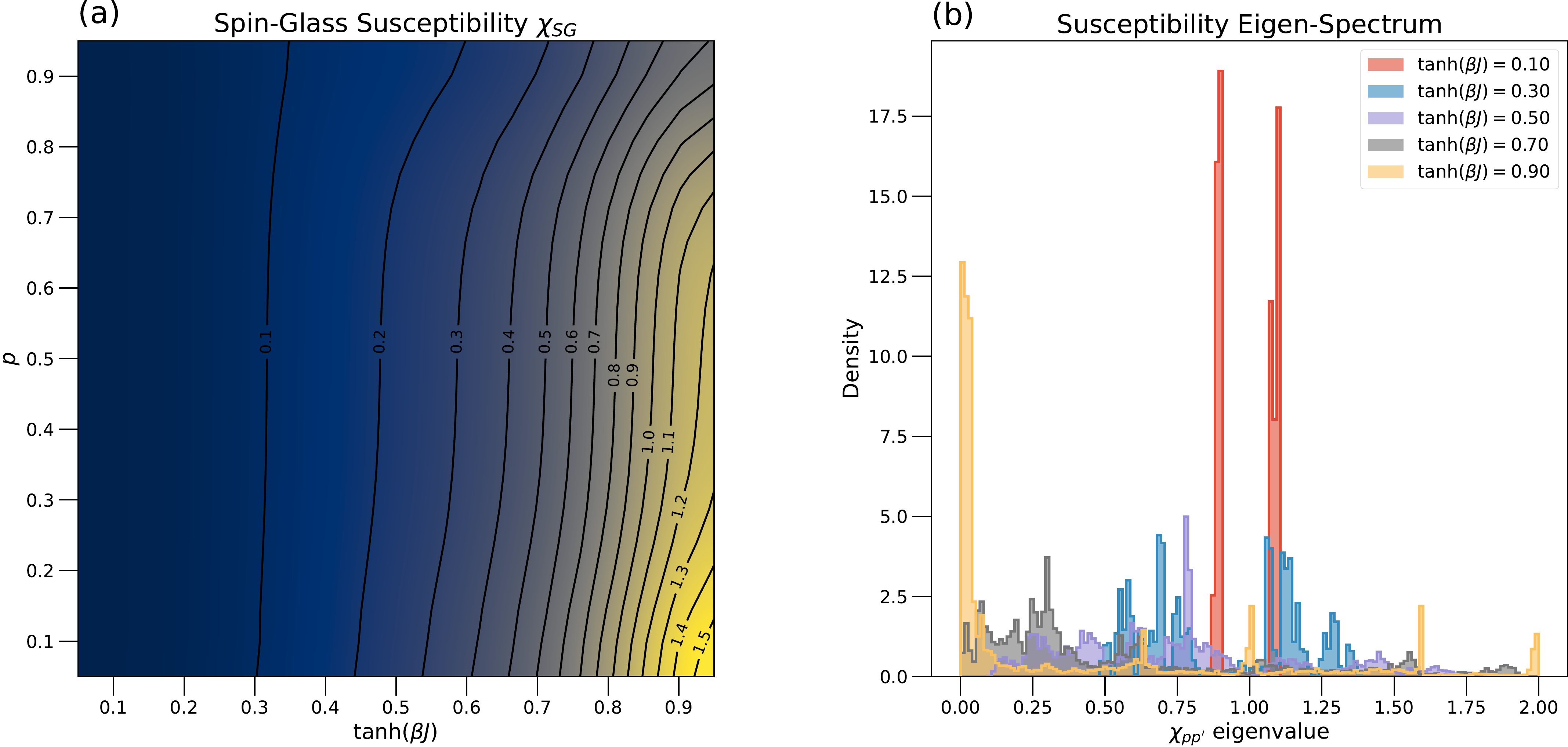}
    \caption{(a) A contour plot of the spin-glass susceptibility $\chi_{SG}$. (b) A histogram of the $\chi_{p p'}$ eigenvalues for the case $p=1/2$. Both plots correspond to a system size of $N=512$ spins ($n=9$) and 100 disorder averages.}
    \label{fig:disordered_model_chi}
\end{figure*}

Lastly, Algorithm~\ref{alg:groundstate} may be used to compute the ground state degeneracy of the model. The degeneracy of the anti-ferromagnetic uniform model, which corresponds to the $p=0$ case of the disordered model considered here, scaled as $d_n \sim A \, b^N$, with $N=2^n$. This suggests that the degeneracy in the disordered case grows as ${\mathbb{E}_J \log_2 d_n = \log_2 A + N \, \log_2 b}$. Fig.~\ref{fig:disordered_model_bfit} depicts the result of performing a standard linear fit to log degeneracy, resulting in an estimate for the base $\hat{b}(p)$. The fit utilized 1000 system samples each for the system sizes given by $n = 3, 4, 5, 6, 7$. The ground state degeneracy decreases as the relative fraction of anti-ferromagnetic couplings decreases, i.e., as $p$ grows. As $p \rightarrow 1$, $\hat{b}(p) \rightarrow$ 1, in agreement with the non-extensive ground state energy observed for the uniform ferromagnetic model. 

\begin{figure*}[ht!]
    \centering
    \includegraphics[width=0.6\textwidth]{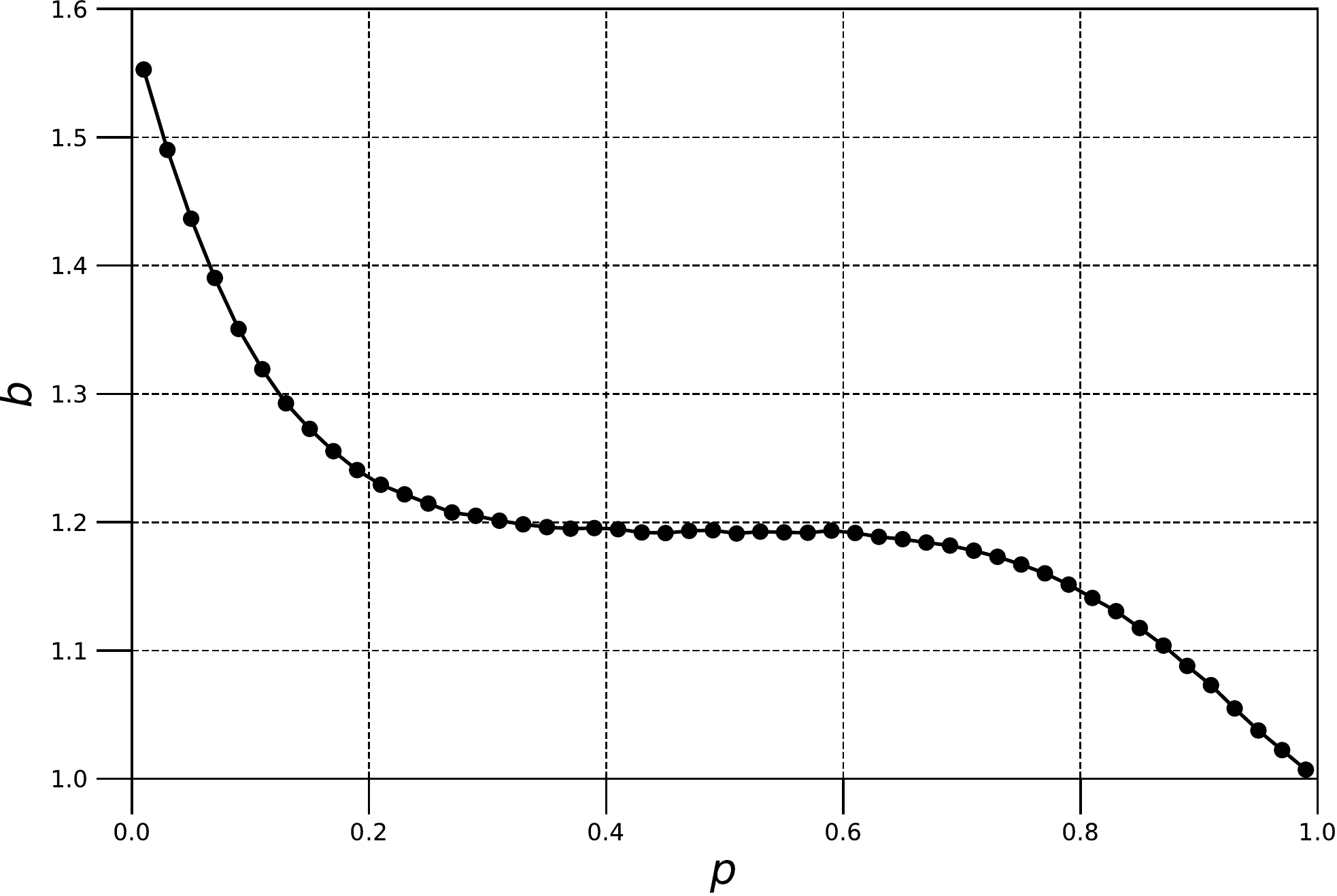}
    \caption{The fitted base $b$ for the exponential growth of the ground state degeneracy as a function of $p$.}
    \label{fig:disordered_model_bfit}
\end{figure*}

\newpage
\bibliography{refs}

\providecommand{\href}[2]{#2}\begingroup\raggedright\begin{thebibliography}{10}

\bibitem{dyson1969existence}
F.~J. Dyson, {\it Existence of a phase-transition in a one-dimensional ising
  ferromagnet},  {\em CMaPh} {\bf 12} (1969), no.~2 91--107.

\bibitem{castellana2013renormalization}
M.~Castellana, {\it The renormalization group for disordered systems},  {\em
  arXiv preprint arXiv:1307.6891} (2013).

\bibitem{castellana2010hierarchical}
M.~Castellana, A.~Decelle, S.~Franz, M.~M{\'e}zard, and G.~Parisi, {\it
  Hierarchical random energy model of a spin glass},  {\em Physical review
  letters} {\bf 104} (2010), no.~12 127206.

\bibitem{castellana2011real}
M.~Castellana, {\it Real-space renormalization group analysis of a
  non-mean-field spin-glass},  {\em EPL (Europhysics Letters)} {\bf 95} (2011),
  no.~4 47014.

\bibitem{franz2009overlap}
S.~Franz, T.~J{\"o}rg, and G.~Parisi, {\it Overlap interfaces in hierarchical
  spin-glass models},  {\em Journal of Statistical Mechanics: Theory and
  Experiment} {\bf 2009} (2009), no.~02 P02002.

\bibitem{castellana2010renormalization}
M.~Castellana and G.~Parisi, {\it Renormalization group computation of the
  critical exponents of hierarchical spin glasses},  {\em Physical Review E}
  {\bf 82} (2010), no.~4 040105.

\bibitem{castellana2011renormalization}
M.~Castellana and G.~Parisi, {\it Renormalization-group computation of the
  critical exponents of hierarchical spin glasses: Large-scale behavior and
  divergence of the correlation length},  {\em Physical Review E} {\bf 83}
  (2011), no.~4 041134.

\bibitem{mezard2009information}
M.~Mezard and A.~Montanari, {\em Information, physics, and computation}.
\newblock Oxford University Press, 2009.

\bibitem{sherrington1975solvable}
D.~Sherrington and S.~Kirkpatrick, {\it Solvable model of a spin-glass},  {\em
  Physical review letters} {\bf 35} (1975), no.~26 1792.

\bibitem{barahona1982computational}
F.~Barahona, {\it On the computational complexity of ising spin glass models},
  {\em Journal of Physics A: Mathematical and General} {\bf 15} (1982), no.~10
  3241.

\bibitem{regge2000combinatorial}
T.~Regge and R.~Zecchina, {\it Combinatorial and topological approach to the 3d
  ising model},  {\em Journal of Physics A: Mathematical and General} {\bf 33}
  (2000), no.~4 741.

\bibitem{galluccio2000new}
A.~Galluccio, M.~Loebl, and J.~Vondr{\'a}k, {\it New algorithm for the ising
  problem: Partition function for finite lattice graphs},  {\em Physical Review
  Letters} {\bf 84} (2000), no.~26 5924.

\bibitem{welsh1990computational}
D.~J. Welsh, {\it The computational complexity of some classical problems from
  statistical physics},  1990.

\bibitem{decelle2014ensemble}
A.~Decelle, G.~Parisi, and J.~Rocchi, {\it Ensemble renormalization group for
  the random-field hierarchical model},  {\em Physical Review E} {\bf 89}
  (2014), no.~3 032132.

\bibitem{jax2018github}
J.~Bradbury, R.~Frostig, P.~Hawkins, M.~J. Johnson, C.~Leary, D.~Maclaurin,
  G.~Necula, A.~Paszke, J.~Vander{P}las, S.~Wanderman-{M}ilne, and Q.~Zhang,
  {\it {JAX}: composable transformations of {P}ython+{N}um{P}y programs},
  2018.

\end{thebibliography}\endgroup
\bibliographystyle{JHEP}

\end{document}